\RequirePackage{amsmath}

\documentclass{nws}
\usepackage[english]{babel}
\usepackage[utf8]{inputenc}
\usepackage{amssymb}
\usepackage{graphicx}

\DeclareMathOperator{\E}{\mathbb{E}}
\DeclareMathOperator{\Var}{Var}
\DeclareMathOperator{\Poisson}{Poisson}

\DeclareMathOperator{\Geo}{Geo}
\DeclareMathOperator{\Exp}{Exp}
\newcommand{\eqD}{\stackrel{\mathcal{D}}{=}} 
\newcommand{\vect}{\boldsymbol} 
\newcommand{\transpose}{^{\top}}

\newcommand\etal{\emph{et~al.}\,}

\newtheorem{theorem}{Theorem}[section]
\newtheorem{lemma}[theorem]{Lemma}

\newenvironment{definition}[1][Definition]{\begin{trivlist}
\item[\hskip \labelsep {\bfseries #1}]}{\end{trivlist}}

\title[Forward Reachable Sets]{Forward Reachable Sets: Analytically derived properties of connected components for dynamic networks}

\author[B. Armbruster, L. Wang, M. Morris]
{Benjamin Armbruster\\
Northwestern University\\
\and ~Li Wang, Martina Morris \\
University of Washington}

\date{Draft updated August 26, 2016}

\begin{document}
\maketitle

\begin{abstract}
Formal analysis of the emergent structural properties of dynamic networks is largely uncharted territory. We focus here on the properties of forward reachable sets (FRS) as a function of the underlying degree distribution and edge duration.  FRS are defined as the set of nodes that can be reached from an initial seed via a path of temporally ordered edges; a natural extension of connected component measures to dynamic networks. Working in a stochastic framework, we derive closed-form expressions for the mean and variance of the exponential growth rate of the FRS for temporal networks with both edge and node dynamics. For networks with node dynamics, we calculate thresholds for the growth of the FRS. The effects of finite population size are explored via simulation and approximation. We examine how these properties vary by edge duration and different cross-sectional degree distributions that characterize a range of scientifically interesting normative outcomes (Poisson and Bernoulli). The size of the forward reachable set gives an upper bound for the epidemic size in disease transmission network models, relating this work to epidemic modeling (Ferguson 2000, Eames 2004).
\end{abstract}

\section{Introduction}

There is growing interest in diffusion over dynamic networks (a recent review may be found in \cite{holme_temporal_2012}), but formal analysis of the emergent structural properties of dynamic networks 
is largely uncharted territory. In this paper we focus on deriving properties of
forward reachable sets (FRS), a natural temporal
extension of the static concept of a connected component. The FRS
is defined as the set of all nodes that can be reached from an origin node
 via a forward reachable path over some period of time;
a forward reachable path is a sequence of temporally
ordered edges that connects two nodes within a specified time interval \cite{moody_importance_2002}. 

Our work is motivated by applications in epidemiology, where an infectious disease passes from one person
to another in a population through a sequence of partnerships that
form, have duration, and dissolve.  Partnerships are different than ``contacts'' in this context; they represent repeated contact with the same person over time, and this temporal clustering of contacts can affect transmission dynamics.  The duration of partnerships influences the stability of the underlying transmission network structure, and the cross-sectional degree distribution (the number of partners at one instant) influences how connectivity emerges over time.  Long-term monogamous partnerships can isolate dyads, reducing epidemic potential, but if partnerships turn over rapidly or people have more than one partner at a time, epidemic potential increases.  The impact of these pair formation features on epidemic outcomes has been recognized in the sexually transmitted infection (STI) literature for decades \cite{dietz_epidemiological_1988,morris_concurrent_1997,kretzschmar_effect_1998, bauch_moment_2000,ferguson_more_2000,leung_si_2014}. 

The FRS can be thought of as the upper bound for the transmission process on a partnership network: an idealized infectious disease for which the probability of transmission is 1, and transmission is instantaneous.   
We model the FRS as an emergent structural property of 
the dynamic network, generated by micro-level processes
that govern node and link dynamics:
the mean cross-sectional degree $k$, the edge dissolution rate $\alpha$, and the node exit rate $\mu$.  Conditional on these parameters, we vary the cross-sectional degree distribution, comparing Bernoulli to Poisson.   These two distributions are chosen because they map to the normative continuum that guides partnership dynamics.  The Bernoulli, at one end, represents serial monogamy, where only one partner at a time is allowed so each partnership is entirely dependent on the presence of another, while the Poisson at the other end assumes all partnerships are formed and dissolved independently, so there is no restriction on the number of partners one can have at the same time.  Populations with the same mean number of partners per person can vary along the continuum defined by these two distributions. Populations with a highly skewed partnership distribution, such as the power-law distribution, are not represented in our analysis. However, long-tailed distributions are usually found for cumulative degree (e.g., number of partners in the last year) \cite{hamilton2008degree}, rather than cross-sectional degree. We will discuss the implications of long-tailed distributions in Section \ref{sec:aggdeg} and in the discussions (Section \ref{sec:degdist}).

Our analysis focuses on the effect of the degree distribution
on the properties of the FRS, conditional on mean degree, edge duration, and node exit rates.  Working in a stochastic framework, we find analytic solutions
for the threshold for growth of the FRS and the mean and variance of the active FRS size over time. Our analysis considers three phases in the evolution
of the FRS: the initial exponential growth phase where we can make
the approximation of an infinite size network, the logistic growth
phase where the finite network effect kicks in, and the equilibrium
where the active FRS is either extinct or varies around a non-zero
size. We evaluate the stochasticity in the FRS by calculating the probability
of extinction, and the variance of the size of the FRS in the initial growth phase.  

Focusing on the FRS allows us to greatly
simplify the pair-formation and moment closure models found in the epidemics literature because we do not track the neighbors' status in our equations.  As a result, we are able to obtain closed-form analytic solutions to several traditional epidemic potential indicators, as a function of basic parameters that can be obtained from egocentrically sampled networks in standard sample surveys.  The methods we develop here can therefore be used to empirically investigate epidemic potential with easily collected network data. 

\section{Definitions and Setup}

\subsection{Precise definitions and assumptions}
\label{sec:precise}

A \textit{dynamic network} can be represented as a list of edges with activities
over intervals (Holme 2012): $\{(n_{1},n_{1}^{\prime},t_{1},t_{1}^{\prime}),\dots,(n_{i},n_{i}^{\prime},t_{i},t_{i}^{\prime}),\ldots\}$,
where each 4-tuple contains the two endpoints of the edge ($n_{i},n_{i}^{\prime}$),
the edge formation time $t_{i}$, and the dissolution time $t_{i}^{\prime}$.
The active intervals ($t_{i},t_{i}^{\prime}$) of an edge must have positive
length ($t_{i}^{\prime}-t_{i}\geq0$), and may not overlap another
interval on the same edge. The networks in our analysis are undirected. 

We also model node entry and exit (analogous to birth and death in
a dynamic population model). Once a node exits the network,
all of its edges dissolve and it can no longer form edges; the exit
is permanent. Each node has only one active lifetime, but each edge may be associated
with multiple active intervals. The two endpoints of an active edge must both be active nodes.

The network is assumed to be homogeneous and in equilibrium. 
All nodes have the same rate of exit $\mu$, regardless of their degree. 
All nodes $i$ with degree $D_i$ have the same edge formation rate $\lambda(D_i)$.
All active edges have the same rate of dissolution $\alpha$. All rates are constant across time. All dynamics are assumed to follow a memoryless stochastic process with the specified rates, which determine how quickly the network is changing. The memoryless process implies that partnership durations are exponentially distributed.

Note that an edge dissolution may be caused by the death of one of the two endpoints, or by a separation process at rate $\sigma$. We can write the dissolution rate as $\alpha=\sigma+2\mu$ and the separation rate as $\sigma=\alpha-2\mu$. Since the separation rate must be positive, there is a constraint that $\alpha/\mu \geq 2$.

We examine dynamic networks with two different cross-sectional degree distributions: Poisson(k) and Bernoulli(p). 
The Poisson degree distribution corresponds to a uniform edge formation process, where $\lambda(D_i)=\lambda_P$ for any node degree, resulting in an Erd\H{o}s-R\'enyi random network at equilibrium when the number of nodes is infinite.  This type of graph is also referred to as a ``simple'' or Bernoulli random graph in different literatures.
The Bernoulli degree distribution restricts edge formation to nodes that have degree 0, meaning 
$\lambda(D_i)=\lambda_B$ if $D_i=0$ and $\lambda(D_i) = 0$ for $D_i\geq 1$. It corresponds to a strictly monogamous social network constraint. Let $k$ be the mean degree of a cross-sectional network taken at equilibrium. Equivalently, $k$ can be the mean instantaneous degree of a node at equilibrium.

\begin{figure}[ht]
    \centering
    \includegraphics[width=12cm]{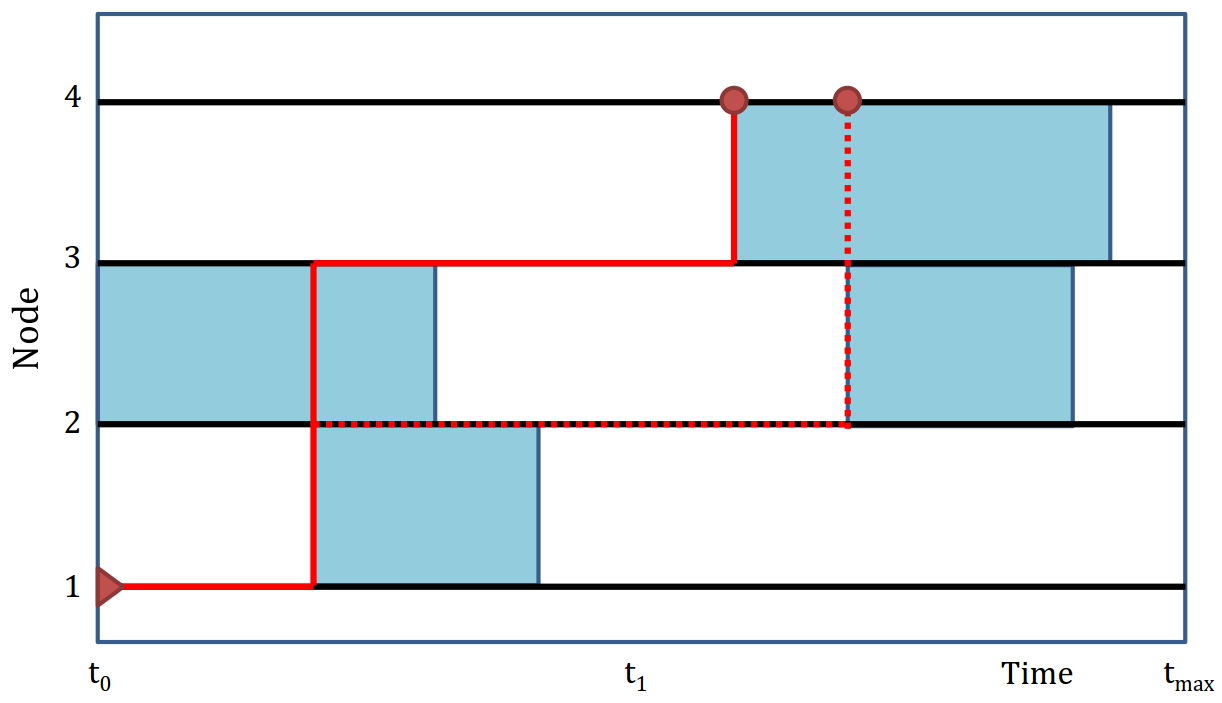}
    \caption{Forward reachable path: the black horizontal lines represent nodes, and the blue shaded regions represent active edges. The red line is a forward reachable path from node 1 to 4, within $[t_0,t_{max}]$. There is a path even though the cross sectional network at $t_1$ has no active connecting edges. The dotted path also connects node 1 to 4, but takes a longer time.
}
    \label{fig:path}
\end{figure}

Given the set of node and edge activities, we can define a forward reachable
path as a time-ordered sequence of active edges between an initial and a destination
node, within a time interval. An equivalent definition using cross-sectional networks
can be found in Nicosia \etal \shortcite{nicosia_components_2012}. Our definition emphasizes
the endpoints of the activity intervals, which are the events that form the
basis of our analysis.

\begin{definition}
There exists a \textit{forward reachable path} from node $i$ to
itself in any time interval (reflexivity). There exists a \textit{forward reachable path}
from $i$ to $k$ in the time interval $(t_{0},t_{max})$
if there exists a node $j$ such that the following two conditions hold:
\begin{itemize}
\item There is a path from $i$ to $j$ in the time interval $(t_{0},t_{j})$, with $t_j\leq t_{max}$.
\item There is an active edge $(j,k,t_{k},t_{k}^{\prime})$ where
$t_{k}<t_{max}$ and $t_{k}^{\prime}\geq t_{j}$.
\end{itemize}
\end{definition}

In our definition for the forward reachable path, there is no delay
between lateral moves in the path, going from one node to another.
In applications like disease transmission or airline connections,
an incubation or transit period $d$ can be specified by requiring
that $t_{k}^{\prime}\geq t_{j}+d$.

Unlike connectivity on a static network, the forward reachable path
is always directed, even if the network is undirected. Connectivity
via forward reachable path is not symmetric: existence of a path from node
$i$ to $k$ does not imply a path from $k$ to $i$.
However, it is transitive if we account for the time coordinate: if
there is a path from $i$ to $j$ in the time interval $(t_{i},t_{j})$
and a path from $j$ to $k$ in the time interval $(t_{j},t_{k})$,
then there is a path from $i$ to $k$ in the interval $(t_{i},t_{k})$. 

\begin{definition}
The \textit{forward reachable set (FRS)}  given $(n_{0},t_{0},t_{max})$ is the set
of all nodes $i$ where a forward reachable path exists between
$n_{0}$ and $i$ in the time interval $(t_{0},t_{max})$. The \textit{active forward reachable set}
$A(n_{0},t_{0},t_{max})$ is the subset of nodes in the FRS that are still active at $t_{max}$.
\end{definition}

One immediate consequence of the reflexivity in forward reachable path
is that if there exists a path between $n_{0}$ and $i$ in the
interval $(t_{0},t_{i})$, then there exists a path between those
two nodes for all intervals $(t_{0},t)$ with $t\geq t_{i}$. So, the forward reachable
set (FRS) is non decreasing in time. However, the active FRS can get smaller
when nodes exit the network, and a node that has exited the network will not be able to form any more
edges, so it will not participate in the growth of the FRS. When all nodes in the active FRS have exited, the active FRS becomes extinct, and its size will stay at zero. 

The dynamics of the network (formation, dissolution, and death) is not influenced by 
whether a node is in the FRS, but the FRS is completely determined by the dynamic network. So,
the FRS can be viewed as a structural property of the dynamic network, analogous
to component size in a static network.

\subsection{Goals for analysis}

Our goal is to model the size of the forward reachable 
set over time on a random dynamic network, as a function of the network
parameters: mean cross-sectional degree $k$, cross-sectional degree
distribution (Bernoulli vs Poisson), edge dissolution rate $\alpha$,
and node exit rate $\mu$ (Table~\ref{table:notation}). We want to define appropriate approximations
and find analytic solutions for the mean and variance of the FRS size,
and compare the results to simulation results.

The networks in the model are finite size, undirected, with population dynamics (node entry and exit). 
The number of active nodes and edges of the network are assumed to be in equilibrium so that there is a balance between the node entry and exit rates, and between
the edge formation rate $\lambda$ and the edge dissolution rate $\alpha$.
This rate balancing will be the foundation of our analysis (Section~\ref{sec:prelim}).

The outcome variables of interest are the expected size and variance
of the active part of the FRS $A$, as a function of time $t_{max}$. 
If there is no node exit $\mu$,
then $A$ is equivalent to the FRS size. The variation is considered
over all initial nodes in the network at $t_{0}$.
There are three
phases in the evolution of the FRS: the initial exponential growth
phase where we can make the approximation of an infinite size network,
the logistic growth phase where the finite network effects kick in,
and the equilibrium where the active FRS either becomes extinct 
or varies in size around a non-zero value.

We split our analysis into three sections.
For the initial exponential growth, we make the approximation that 
a randomly selected node is already in the FRS with probability 0; this is equivalent
to assuming an infinite network size. Given this assumption we can
derive exact solutions for the expectation and variance of the FRS size
(section~\ref{sec:exponential}). For the logistic growth phase, we modify our equations
to include the finite network size effect, using a mean-field approximation
(section~\ref{sec:logistic}). For the equilibrium persistence, we use the
mean-field approximation to derive the expected fraction of the network
in the active FRS ($A$) (section~\ref{sec:equilibrium}). We use simulations to show the
probability of extinction in this stochastic process, and
derive this probability analytically. Our main results will compare the behavior
of the active FRS between Bernoulli and Poisson degree distributions,
for various values of mean degree, dissolution rate, and node exit rate.
This comparison has strong implications for the role of concurrent
partnerships in the spread of a sexually transmitted disease in a
sexual network (section~\ref{sec:discussion}).

\section{Preliminary Analysis}
\label{sec:prelim}

\begin{table}[ht]
\caption{Parameters and notation}
\begin{tabular}{cl}
\hline 
~Parameter~ & Definition\tabularnewline
\hline 
$k$ & Mean instantaneous degree\tabularnewline
$\alpha$ & Edge dissolution rate, per edge\tabularnewline
$\mu$ & Node exit rate, per active node\tabularnewline
$\sigma$ & Edge separation rate, per edge (derived)\tabularnewline
$\lambda$ & Edge formation rate, per eligible node (derived)\tabularnewline
$g$ & FRS growth rate (derived)\tabularnewline
\hline 
\end{tabular}
\label{table:notation}
\end{table}

In this section we discuss the preliminaries for the derivation of the equations for the FRS growth in Section~\ref{sec:exponential}.  These include calculating the size of the connected components; relating the rate of partnership formation to the network parameters; and deriving how the cumulative or ``aggregate'' degree of a node grows over time. A more rigorous derivation of these properties can be found in \cite{leung_dynamic_2012}.

Let $C$ be a random variable denoting the size of the connected component attached to a randomly chosen node, on a cross-sectional (static) network.  For a Bernoulli degree distribution, clearly $\E[C]=1+k$.  For Poisson \shortcite{newman_component_2007}, when $k<1$
\begin{align}
\E[C] &= (1-k)^{-1}\\
\E[C^2] &= (1-k)^{-3}\\
\Var[C] &= k (1-k)^{-3}
\end{align}

We do not consider the $k\geq 1$ case since it would lead to a giant connected component, with $\E[C]=\infty$. These results are valid on a homogeneous network of infinite size. For finite or non-homogeneous networks, the component size can be obtained via simulation \cite{morris_microsimulation_2000}.

The total dissolution rate $\alpha$ is a combination of partnership separation $\sigma$ and death of either partner, so $\alpha=\sigma+2\mu$, and the expected partnership duration is $1/\alpha$. The separation rate is $\sigma=\alpha-2\mu$.

Next we derive the partnership formation rates per eligible node, $\lambda_P$ and $\lambda_B$ for the Poisson and Bernoulli cases respectively, given the average degree $k$, edge dissolution rate $\alpha$, and node exit rate $\mu$. We focus on the degree $D(t)$ over time of a random active node, and model $D(t)$ as a Markov birth-death process on $\{0,1,\ldots\}$ in the Poisson case and on $\{0,1\}$ in the Bernoulli case. Since the degree $D(t)$ is only relevant when the node is alive, we are implicitly conditioning on that, so the death rate of the birth-death process (not the death rate of the node) is $(\sigma+\mu)D(t)$: an active node may lose a partner due to separation at rate $\sigma$, or partner death at rate $\mu$.

In steady-state, $\E[D(t)]=k$, balancing the birth and death rates of the birth-death process. For the Poisson case where all nodes are eligible for formation, this balance is $\lambda_P=(\sigma+\mu) k$. For the Bernoulli case, only the degree-0 nodes are eligible to form an edge, and only degree-1 nodes can dissolve an edge. Noting that fraction of nodes with degree 1 is $k$, and the fraction of nodes with degree 0 is $1-k$, the equilibrium balance is $(\sigma+\mu)k=\lambda_B (1-k)$ and therefore, $\lambda_B=(\sigma+\mu)k/(1-k)$.  

\subsection{Cumulative degree}\label{sec:aggdeg}
As a final piece of the preliminaries we look at the aggregate or cumulative degree. Often, surveys do not ask ``how many partners do you have \emph{at this moment}?'' but commonly ``how many partners (or new partners) have you had in [time period]?''.  The first is $D_{eq}$, the equilibrium distribution of $D(t)$. Here we examine on the second. Let $D_a(t)$ be the cumulative number of partners of some node from its birth until a duration $t$ later.  In the Poisson case, the node acquires partners independent of its current number of partners, so $D_a(t)-D_a(0)$ is just a Poisson process with rate $\lambda_P$.  The Bernoulli case is slightly trickier since the time between new partners (i.e., when $D_a$ increases) is $\Exp(\lambda_B)+\Exp(\sigma+\mu)$.  Note that the expectation of this time is still $1/\lambda_P$ as in the Poisson case, but the variance of this time is now $\lambda_P^{-2} (1-2k(1-k))$ which is smaller than the variance in the Poisson case, $\lambda_P^{-2}$.  Essentially, $D_a$ is a renewal process, specifically an equilibrium renewal process if $D_a(0)\eqD D_{eq}$.  Thus, 
\begin{equation}
	\E[D_a(t)-D_a(0)]=\lambda_P t
\end{equation}
in both the Poisson and Bernoulli cases and 
\begin{equation}
	\Var[D_a(t)-D_a(0)]=\lambda_P t
\end{equation}
in the Poisson case and asymptotically in the Bernoulli case, 
\begin{equation}
	\Var[D_a(t)-D_a(0)]\sim (1-2k(1-k))\lambda_P t
\end{equation}
using standard renewal results.  Note that both grow linearly with time, the expectations are the same for Poisson and Bernoulli, and the variance is smaller for the Bernoulli case.

For the Poisson case, the number of partners within a time interval would have a Poisson distribution, if the node is actively acquiring partners for the whole interval. If however the node became active during this period, then the number of partners within the time interval would be less. This heterogeneity would lead to a more skewed distribution for cumulative degree in our model.

One related calculation that is simple to perform is the lifetime number of partners, $D_a(L)$.  In the Poisson case, the number has a Geometric distribution: 
\begin{equation}\label{eq:DaLPoisson}
	D_a(L)-D_a(0)\sim\Poisson(\lambda_P L)=\Geo_0(1/(1+\lambda_P/\mu))
\end{equation}
due to standard properties of these distributions. In either the Poisson or Bernoulli case we can use iterated expectations to show that $\E[D_a(L)-D_a(0)]=\E[\lambda_P L]=(\alpha/\mu-1)k$.  Thus, 
\begin{equation}
	\E[D_a(L)]=k\alpha/\mu,
\end{equation}
assuming births have an equilibrium number of partners. Poisson and Bernoulli degree networks have the same expected lifetime number of partners.

\subsection{Simulation setup}
\label{sec:sim}

We will compare our analytic results to simulations on random
dynamic networks, both to verify the analysis and to show
differences between the theory and the actual stochastic
process. The networks we simulated have 500 nodes over a range of
parameters $k$, $\alpha$, and $\mu$. 
The random dynamic networks are simulated as continuous time Markov processes with three types of events,
using the stochastic simulation algorithm \cite{Banks11simulationalgorithms}. 
A formation event will create an active edge between the sender node and a random node
in the set of eligible receiver nodes (all other nodes for Poisson, all nodes with degree zero for Bernoulli). Edge separation is defined as the dissolution of an edge that's not due to the exit of one of the nodes. A separation event deactivates a random active edge in the network.
A node exit event deactivates the node and all the edges incident on it.
After the network simulation is completed, we calculate the active
forward reachable set ($A(t,i)$) for every active node at $t=0$. 
We used the R package ``tsna'' for this calculation \cite{bender-demoll_tsna:_2015}.

In our simulations, a node that exits the network is immediately
replaced by a new node (``reincarnation''). The new node enters
the network with degree 0. This is significant for a social network
interpretation, but it adds an inhomogeneity to the dynamic network. The
new nodes do not have the same degree distribution as the rest of
the network, which can affect the equilibrium size of the forward
reachable set.  However, new nodes do not affect the exponential growth phase: even if
they enter the network with non-zero degree, it is unlikely that they are in contact with
the small initial FRS at birth. In Appendix B (Section~\ref{sec:meanfield_adjustment}), we will modify our model to account for the fact that the newly born nodes have degree 0. 

To match the simulated network to the theoretical parameters, we adjust
the rate of edge formation of each active node $\lambda$ and the separation
rate $\sigma$. To match the mean degree $k$, we set the formation
rate for Poisson degree distribution to $\lambda=k\alpha/2$,
and for Bernoulli to $\lambda=\frac{1}{2}\alpha k/(1-k)$ for ``eligible'' nodes of degree 0. 
The total dissolution rate is split
up into two processes, separation and node exit. The separation rate
is set as $\sigma=\alpha-2\mu$.

\begin{figure}[ht]
    \centering
    \includegraphics[width=\textwidth]{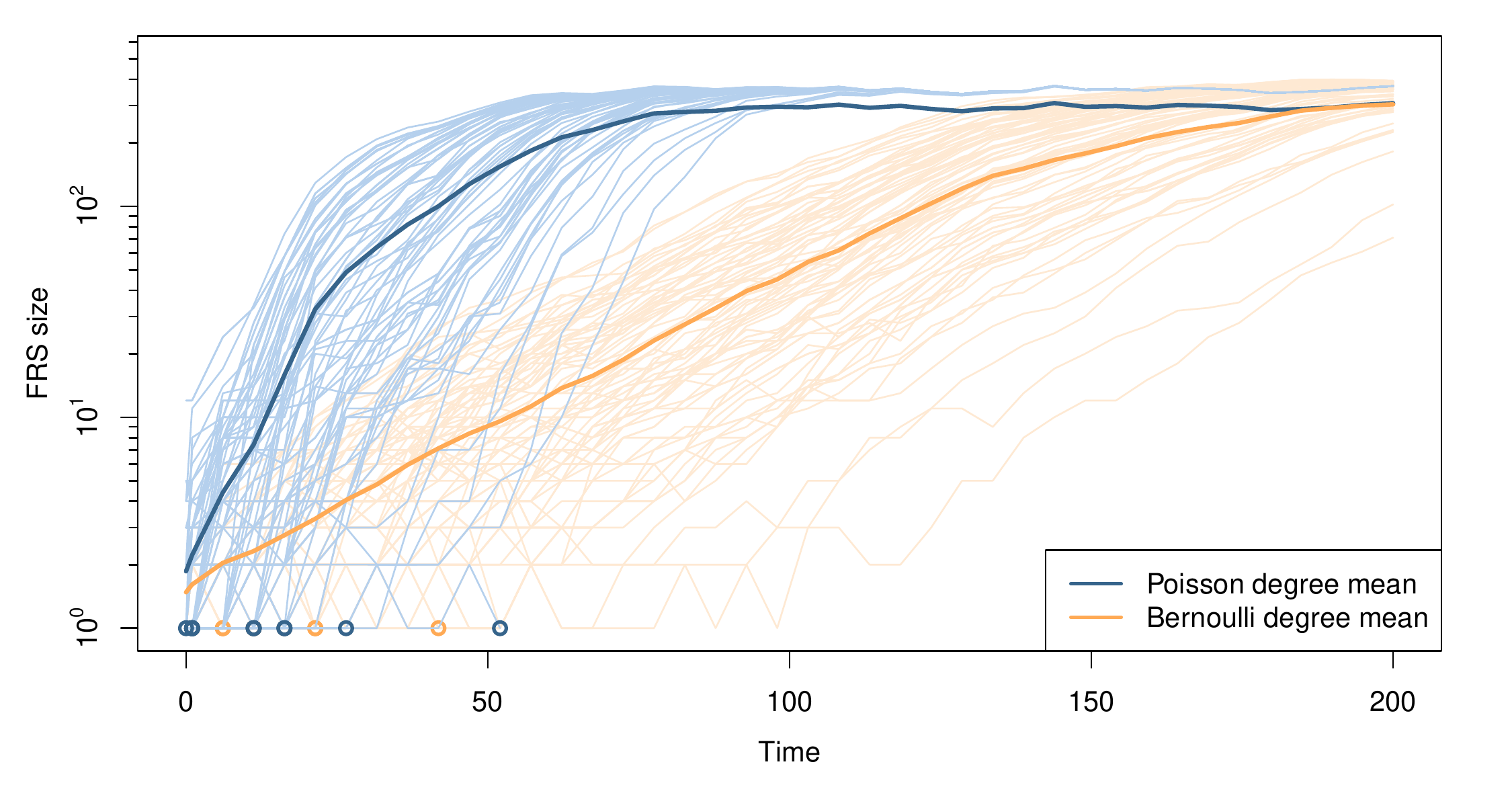}
    \caption{Growth of
the active forward reachable set (FRS) for network parameters $k=0.5$,
$\alpha=1/6$, $\mu=1/40$, and network size 500. The thin blue lines show
the FRS trajectories starting at different initial nodes on a single simulated network
with Poisson degree distribution; the lighter orange lines, for Bernoulli. The
circles along the bottom represent FRS extinctions, and the thick lines are for the corresponding mean FRS sizes.}
    \label{fig:traceplot}
\end{figure}

Figure~\ref{fig:traceplot}  
shows realizations of the growth of the FRS over time,
for different initial nodes on a single dynamic network simulation.
We can separate the growth dynamics into three phases.
When the FRS size is small compared to the network, it exhibits exponential
growth. When the active FRS size is the same order of magnitude as
the network size, the growth curve becomes logistic and eventually hits a plateau. 
Although there is much variation in the time for the FRS size to take-off,
eventually all the realizations starting at different nodes converge to an active FRS of the same size.

\section{Exponential growth phase}
\label{sec:exponential}

For the exponential growth phase, the size of the forward reachable set is small compared to the size of the network. In the epidemiology context, this means the fraction of susceptibles in the population is close to 1 \cite{britton_stochastic_2014}. We will make the approximation that when an edge forms on a node in the FRS, it connects to a node not in the FRS with probability 1.  This approximation is equivalent to having an infinite size network. 

\subsection{General ODE Model}\label{sec:generalODE}
Our general approach models the forward reachable set as a continuous time Markov chain (CTMC).
Consider a vector-valued Markov jump process with state variables $\vect{V}(t)$ and transition rate matrix $Q$. We assume all vectors are column vectors by default.
The Kolmogorov backward equations \cite{karlin_first_2014} state that for any function $b(\vect{V})$,
\begin{equation}
	\E[b(\vect{V}(t))]'=\E[c(\vect{V}(t))],
\end{equation}
where we define 
\begin{align}
	c(\vect{v})&=\sum_{\vect{w}} Q(\vect{v},\vect{w}) b(\vect{w})\\
    &=\sum_{\vect{w}\neq \vect{v}} Q(\vect{v},\vect{w}) (b(\vect{w})-b(\vect{v})).
\end{align}
The last equality is due to the $\sum_{\vect{w}\neq \vect{v}} Q(\vect{v},\vect{w})=-Q(\vect{v},\vect{v})$ property of transition matrices.

The processes of interest to us have special structure.  
They only have a few different types of jumps (i.e., transitions) and the jumps are of the following form.
In state $\vect{v}$ we take jumps of type $i$ at a rate linear in the state, 
$\vect{h}_i\transpose \vect{v}$, where $\vect{h}_i$ is a constant.
Furthermore, a jump of type $i$ transitions the state from $\vect{v}$ to $\vect{v}+\Delta \vect{V}_i$ where $\Delta \vect{V}_i$ may be random but is required to be independent of $\vect{v}$.

With these assumptions and choosing the identity for $a(\vect{v})$, we can write the backward equations as
\begin{equation}\label{eq:mastereq}
\E[\vect{V}]'=\sum_i \E[\Delta \vect{V}_i (\vect{h}_i\transpose \vect{V})]=\vect{H}\E[\vect{V}],
\end{equation}
where matrix 
\begin{equation}
	\vect{H}=\sum_i \E[\Delta \vect{V}_i]\vect{h}_i\transpose.
\end{equation}
In the next sections we specify how our problem can be written in this form
deriving the equations describing how the expected size of the FRS grows over time, first for the Poisson case, then for the Bernoulli case, before finally comparing the results.  Derivations and results for the dynamics of the variance of the size of the FRS are in appendix A.

\subsection{Poisson degree networks}\label{sec:PoissonEqs}
For the Poisson case, the state variables we need are $A$, the size of the active FRS, and $R$, the number of inactive (i.e., dead) nodes in the FRS, so in the vector notation from above,  $\vect{V}=(A,R)$. There are two types of transitions. The first is the formation of an edge where one endpoint is in the FRS (\textit{formation}). This happens at a rate of $\lambda_P A$ and adds a connected component to the active FRS ($A\to A+C$ ). 
Thus using the vector notation, $\vect{h}_1=(\lambda_P,0)$ and $\Delta \vect{V}_1=(C,0)$. The second type of transition is node exit (\textit{death}).
It happens at a rate of $\mu A$, and removes a single node from the FRS ($A\to A-1$ and $R\to R+1$ ).  Using the vector notation, $\vect{h}_2=(\mu,0)$ and $\Delta \vect{V}_2=(-1,1)$.

Thus, equation (\ref{eq:mastereq}) gives us the following differential equations,
\begin{align}
\E[A]' &=(\lambda_P \E[C]-\mu)\E[A]  \label{eq:pois}\\
\E[R]' &=\mu\E[A].
\end{align}
At time 0, $R=0$ and $A$ equals the size of the connected component of the initial node of the FRS.  If the initial node is chosen at random, then $\E[A(0)]=\E[C]$. Node birth events are not counted here because new nodes have degree 0 and do not add to the active FRS.

These differential equations are linear with constant coefficients, and thus have simple exponential solutions:
\begin{align}
\E[A(t)] &=\E[A(0)] e^{g_P t},\\
\E[R(t)] &=(e^{g_P t}-1) \E[A(0)] \mu/g_P,
\end{align}
where the exponential growth rate $g_P=\alpha (k-x)/(1-k)$, with $x=\mu/\alpha$.
This makes sense from a dimensional analysis perspective: since $g$ is a growth rate, it must be the product of a rate parameter such as $\alpha$ and a function of the dimensionless constants, $k$ and $x$.

\subsection{Bernoulli degree networks}\label{sec:BernoulliEqs}
For the Bernoulli case, in addition to $A$ and $R$, we must add a variable $W$ to keep track of the number of degree-0 nodes in the FRS. Only nodes in $W$ are eligible to form active edges. The vector of state variables is now $\vect{V}=(A,W,R)$.  Our Markov process has the following types of transitions:
\begin{itemize}
\item Formation. $A\to A+1$ and $W\to W-1$ at rate $W\lambda_B$.\\
$\Delta\vect{V}_1=(1,-1,0)$, $\vect{h}_1=(0,\lambda_B,0)$.
\item Dissolution. $W\to W+2$ at rate $\sigma(A-W)/2$.  Here $(A-W)/2$ counts the number of partnerships in the FRS.\\
$\Delta\vect{V}_2=(0,2,0)$, $\vect{h}_2=(\sigma/2,-\sigma/2,0)$.
\item Death of a degree-0 node in the FRS. $A\to A-1$, $W\to W-1$, and $R\to R+1$ at rate $\mu W$.\\
$\Delta\vect{V}_3=(-1,-1,1)$, $\vect{h}_3=(0,\mu,0)$.
\item Death of a degree-1 node in the FRS. $A\to A-1$, $W\to W+1$, and $R\to R+1$ at rate $\mu (A-W)$.\\
$\Delta\vect{V}_4=(-1,1,1)$, $\vect{h}_4=(\mu,-\mu,0)$.
\end{itemize}

This leads to the following set of differential equations:
\begin{align}
\E[A]' &= \lambda_B\E[W]-\mu\E[A]\\
\E[W]' &= -(\lambda_B+\mu)\E[W]+(\sigma+\mu)(\E[A]-\E[W])\\
\E[R]' &= \mu\E[A].
\end{align}

As in the Poisson case, the equation for $\E[R]'$ decouples from the rest; none of the other equations rely on $\E[R]$.  Thus we end up with a two dimensional linear system with constant coefficients, where the growth rate of the solution is determined by the largest eigenvalue of the coefficient matrix.  The eigenvalues are $(\alpha-\mu)(\pm\sqrt{1+4k(1-k)}-1)/(2(1-k))-\mu$. Since $1-k>0$ and $\alpha-\mu>0$, the growth rate is:
\begin{align*}
g_{B} &=  (\alpha-\mu)\frac{\sqrt{1+4k(1-k)}-1}{2(1-k)}-\mu\\
 &=  \alpha\left(\frac{\sqrt{1+4k(1-k)}-1}{2(1-k)}(1-x)-x\right)
\end{align*}
where again $x=\mu/\alpha$. Thus $\E[A(t)]=O(e^{g_B t})$ and the same for $\E[R]$.  In fact in both the Poisson and Bernoulli cases, the differential equation for $\E[R]$ is the same and thus, $\E[R(t)]\sim \E[A(t)]\mu/g$ as $t\to\infty$.  This growth rate has simple edge cases: $g_B=-\mu$ for $k=0$ and $g_B=\sigma$ for $k=1$.

\subsection{Key results}
\label{sec:keyresults}
Since $k$ and $x$ are both dimensionless,
dimensional analysis tells us that we can write the growth rate in the form $g=\alpha f(k,x)$.  Equivalently, we can arbitrarily choose $\alpha=1$ and plot $g$ as a function of $k$ for various choices of $\mu$.  Figure \ref{fig:1} shows this.  From this we see that the Poisson growth rates go to infinity as $k\to 1$, because most nodes are connected as part of a giant component in the cross-sectional network. The Bernoulli degree network has a max component size of 2, so its FRS growth is limited by the edge dissolution rate. We can show that the growth rate for Poisson degree distributed network is faster than Bernoulli, $g_P\geq g_B$, given any parameter values. 

\begin{figure}[h]
    \centering
    \includegraphics[width=3.2in]{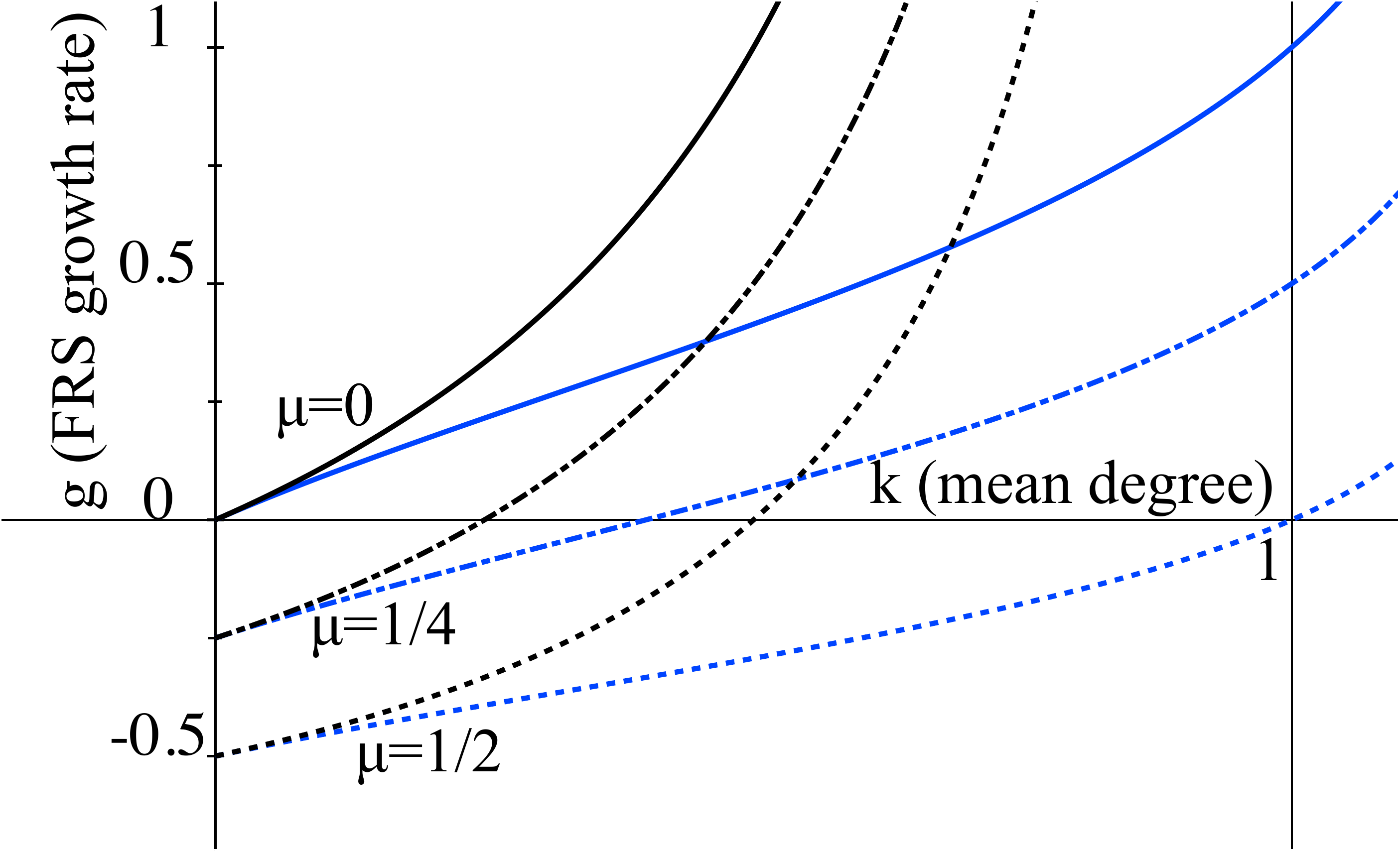}
    \caption{Growth rate, $g$. Black lines are for the Poisson case, blue lines for the Bernoulli case. Here $\alpha=1$ and we show $\mu=0,1/4,1/2$.  For other $\alpha$ values, scale $g$ and $\mu$ accordingly.} 
    \label{fig:1}
\end{figure}

\begin{lemma}\label{lem:gPgB} $g_P\geq g_B$ for all parameter values.
\end{lemma}
\begin{proof}
Let $x=\mu/\alpha$.  Then $(g_P-g_B)2(1-k)/(\alpha(1-x))=2k+1-y$ where $y=\sqrt{1+4k(1-k)}$.  Now $(2k+1)^2=4k^2+4k+1\geq -4k^2+4k+1=y^2$, proving the claim.
\end{proof}

In simulations, the exponential growth rate of the expected FRS size matches
up well with the analytic results Figure~\ref{fig:traceplot}. 
The growth rate is higher for larger cross-sectional
mean degree $k$ and larger values of the dimensionless parameter
$\alpha/\mu$, which is proportional to the number of formations per
node lifetime (Section~\ref{sec:prelim}). We can make the comparison between Poisson
and Bernoulli degree distributions for different values of $k$ and
$\alpha/\mu$. Figure~\ref{fig:circleplot} shows that the FRS growth rate
is always larger for a Poisson degree network compared to a Bernoulli
degree network, given the same parameter values. A Poisson degree
network can achieve the same growth rate with a smaller mean degree,
or fewer formations per node lifetime. 

\begin{figure}[ht]
    \centering
    \includegraphics[width=10cm]{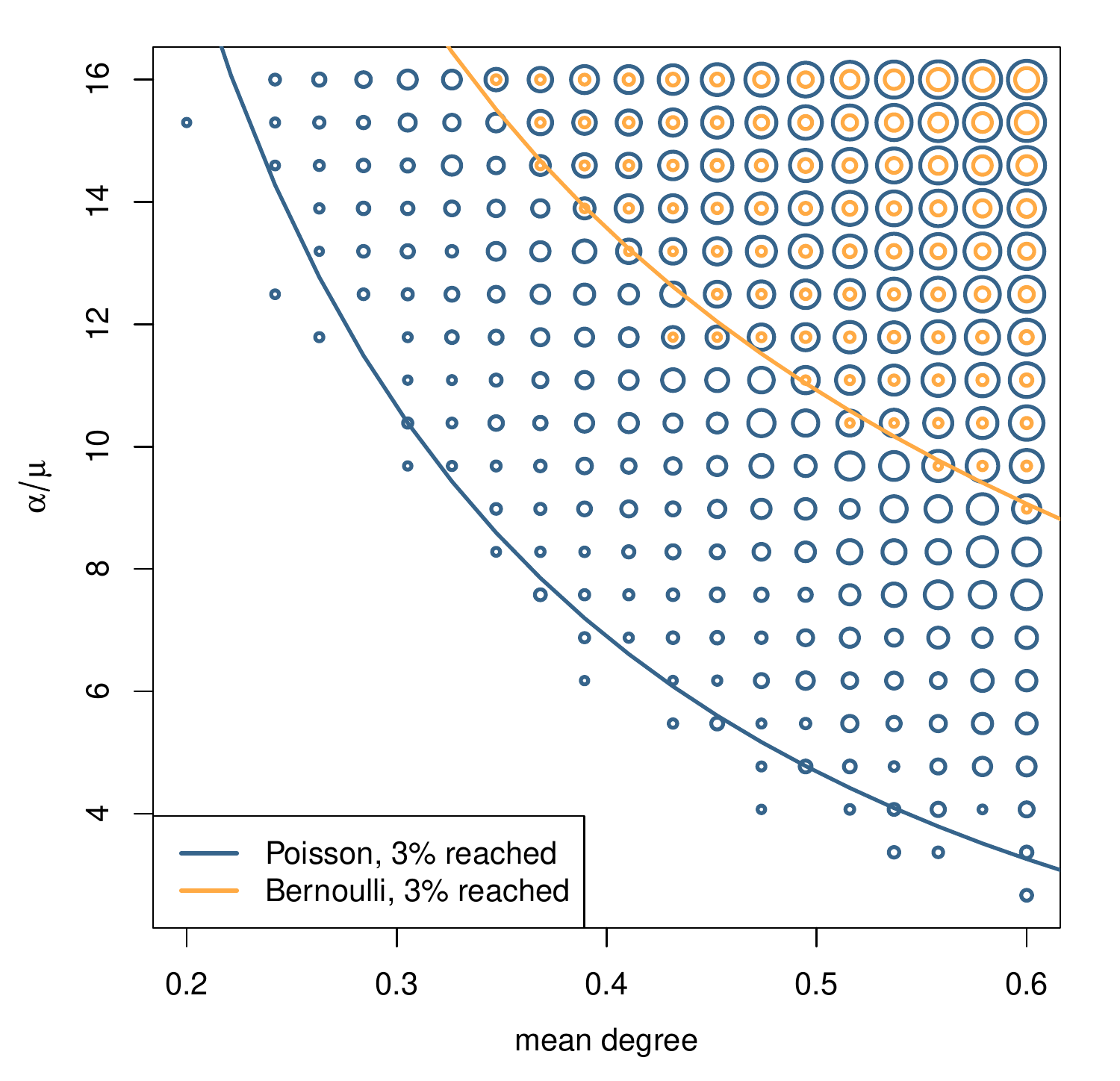}
    \caption{Comparison
of the FRS growth rate across simulated networks with different parameter values. The size
of the circles represent the mean FRS size at $t=30$, which we can
us as a proxy measure of the exponential growth rate. We only show
circles for FRS sizes $>$3\%. The two lines highlight the theoretic 
contours of equal growth rate.}
    \label{fig:circleplot}
\end{figure}

We can define a growth threshold, $g>0$, that separates the parameter values where the FRS grows exponentially from those where it shrinks or doesn’t grow.  From dimensional analysis we can write the threshold as $f(k)>\mu/\alpha$ for some function $f$ or equivalently $R:=f(k)\alpha / \mu> 1$ to write it in a form similar to the standard form for $R_0$.  In the Poisson case, $g_P=\alpha (k-x)/(1-k)$, so $g_P>0$ is equivalent to $k>\mu/\alpha$.  For the Bernoulli case we have $1-2(1-k)/(\sqrt{1+4k(1-k)}+1-2k)>\mu/\alpha$.  The thresholds are similar for small $k$ because for small values of $k$, the left hand side of the Bernoulli threshold is $k+o(k)$, matching the Poisson threshold. In addition Lemma \ref{lem:gPgB} of course implies that the Bernoulli threshold is always greater than the Poisson one (i.e., $g_B>0$ implies $g_P>0$).  Figure \ref{fig:2} plots the two thresholds as functions of $k$ and the dimensionless quantity $\alpha/\mu$.

\begin{figure}[h]
    \centering
    \includegraphics[width=3.2in]{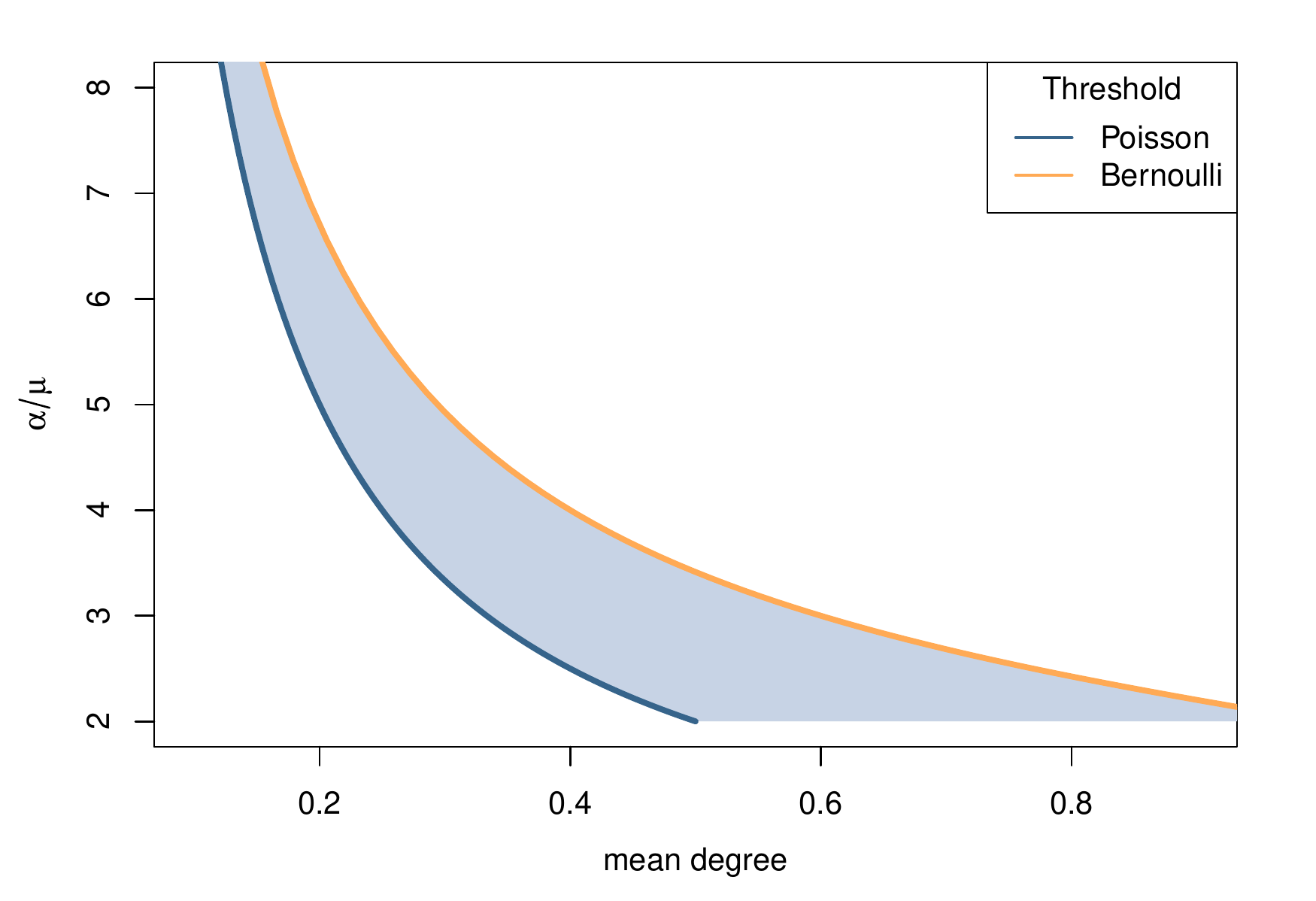}
    \caption{Growth threshold. Above the threshold, growth is positive.  There is a range of parameters (the blue shaded region) with positive growth for Poisson, but not for Bernoulli degree networks. Note $\alpha/\mu<2$ is not possible (Section~\ref{sec:precise}).} 
    \label{fig:2}
\end{figure}

\section{Logistic growth phase}
\label{sec:logistic}

When the size of the forward reachable set grows to a significant fraction (about 10\%) of network size, the assumptions we used for the exponential growth phase are no longer valid. When a new edge is formed, the probability that the node in the FRS connects to a node not already in the FRS is no longer approximately 1; it is reduced to $1-E[A]/n$. This saturation effect is captured by the mean-field models that modify the differential equations from Section~\ref{sec:exponential}.

For Poisson degree networks, the differential equation becomes 
\begin{equation}
\E[A]'=\E[A](1-\E[A]/n)\lambda_P E[C]-\mu \E[A].\label{eq:PmeanfieldA}
\end{equation}

For the Bernoulli case, we get
\begin{align}
\E[A]' &= (1-\E[A]/n)\lambda_B \E[W] - \mu\E[A]\label{eq:BmeanfieldA}\\
\E[W]' &= -(\lambda_B+\mu)\E[W] + (\sigma+\mu)(\E[A]-\E[W]).\label{eq:BmeanfieldW}
\end{align}

We can solve the equations above numerically to approximate
the expected FRS growth in the logistic phase. The growth rate decreases as the FRS gets larger, but it never drops below 0. Figure~\ref{fig:traceplot2} shows how the average FRS size in simulations matches up with the mean-field solutions for a 500 node network, up until the equilibrium.

\begin{figure}[ht]
    \centering
    \includegraphics[width=12cm]{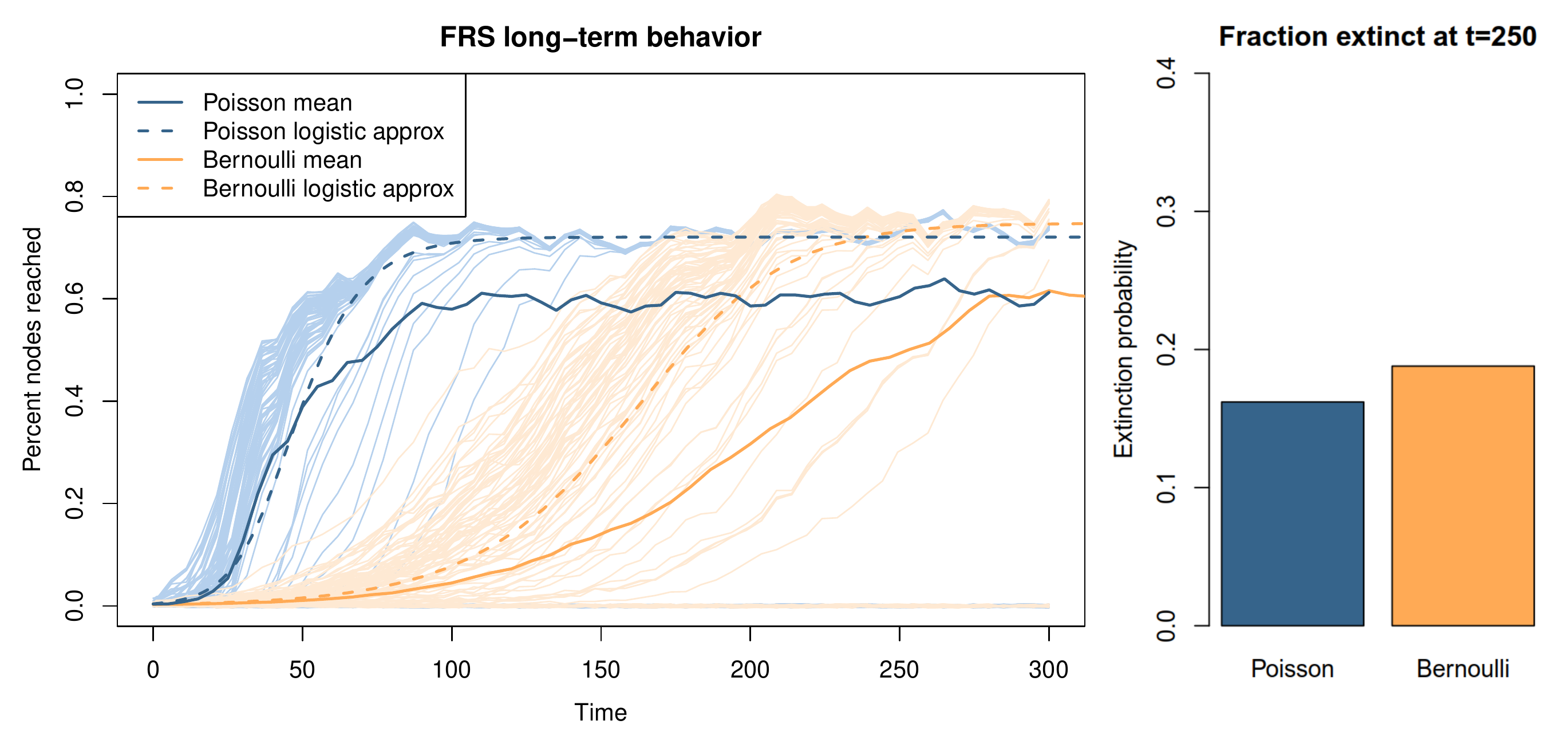}
    \caption{Long-term behavior
of the active forward reachable set for a simulated network with parameters $k=0.5,$
$\alpha=1/6$, $\mu=1/40$, and network size 500 (same as in Figure \ref{fig:traceplot}). The non-logged y-axis
 shows the logistic growth curve; the mean-field logistic approximations match up with the persistent FRS size. 
 The different trajectories show the FRS growth for different starting nodes.
 At equilibrium, the FRS size distribution is bimodal, with the lower mode showing the extinction probability (right panel). All the trajectories that persist eventually converge to an active FRS of the same size.}
    \label{fig:traceplot2}
\end{figure}

Setting the right hand side
to zero allows us to solve for an approximation to the equilibrium expected size.
Figure~\ref{fig:3} shows the dependence of the mean-field equilibrium on the parameters.

\begin{figure}[ht]
    \centering
    \includegraphics[width=3.2in]{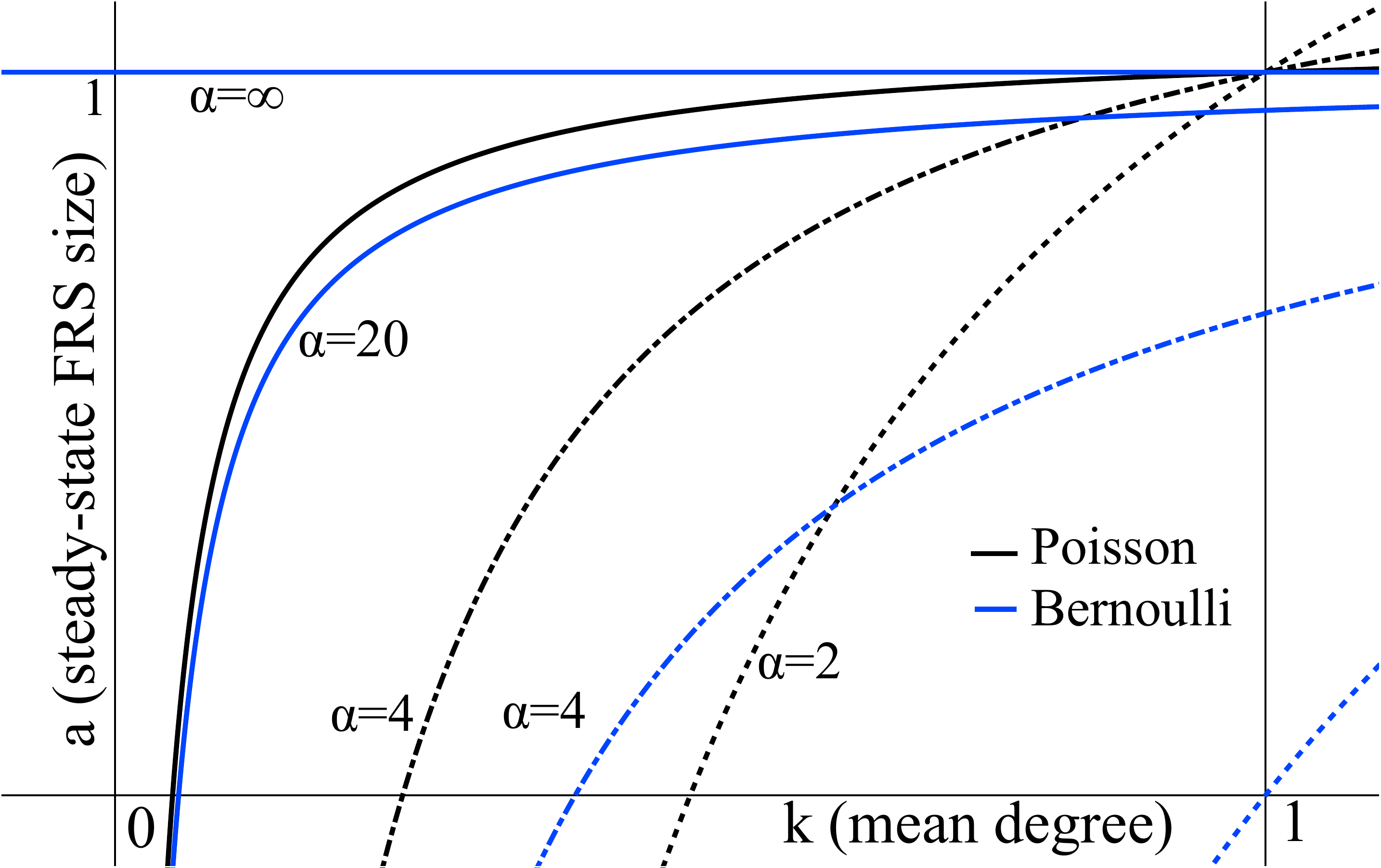}
    \caption{Mean-field equilibrium forward reachable set size for $\mu=1$.  Black lines are for the Poisson case and blue lines for the Bernoulli case.  Lines are shown for $\alpha=2,4,20,\infty$.}
    \label{fig:3}
\end{figure}

\section{Equilibrium phase}
\label{sec:equilibrium}

At equilibrium, the FRS size distribution becomes bimodal (Figure~\ref{fig:traceplot2}).
For some of the initial nodes, the FRS takes off and reaches a persistent
equilibrium value. For other initial nodes, the FRS shortly becomes
extinct, and will remain at zero, even when the network parameters
are above the threshold for growth. Our mean-field model is deterministic
and does not account for extinctions of the FRS, which happens stochastically. 
The mean-field solutions match up with the average
persistent FRS size (for the set of initial nodes that do not have their FRS go extinct).

\begin{figure}[ht]
    \centering
    \includegraphics[width=10cm]{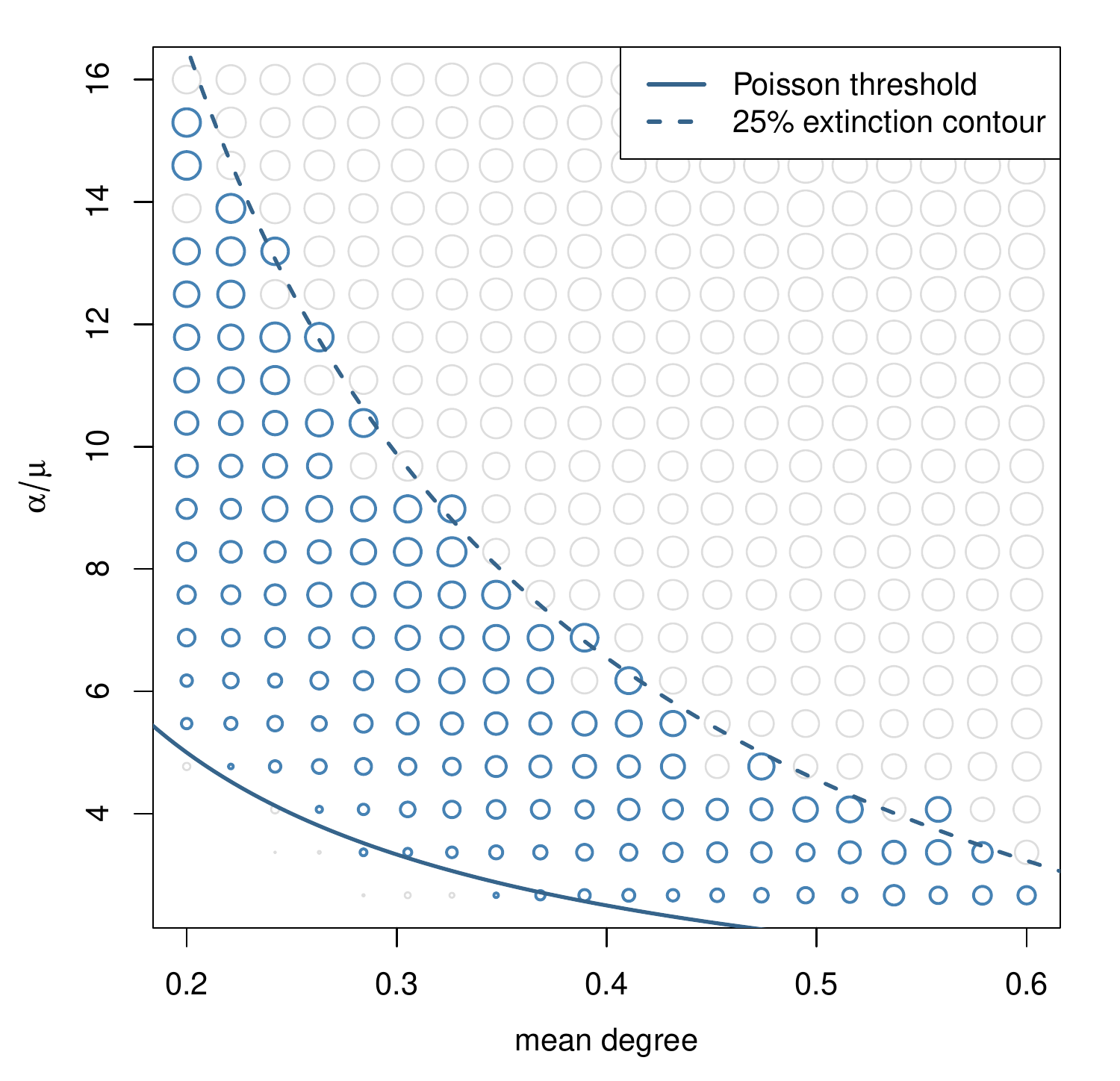}
    \caption{The persistent
FRS size, plotted as circle size, across different network parameter values at equilibrium.
The band of blue circles shows the region where extinction
probability is between 25\% and 75\%. In the lower left corner, the
FRS goes extinct for most initial nodes. In the upper right, the FRS
will grow to the persistent size most of the time. }
    \label{fig:equilplot}
\end{figure}

\subsection{Persistent FRS size}
\label{sec:persistentsize}
The logistic mean-field equations from Section~\ref{sec:logistic} allow us to calculate the equilibrium FRS fraction for the Poisson case, by setting the growth rate to zero:
\begin{equation}
	a_P = \E[A]/n = 1-\mu/(\lambda_P \E[C]) = (k-x)/(k(1-x)),
\end{equation}
where $x=\mu/\alpha$.  Note that the threshold for a positive steady state, $a_P>0$, coincides with the threshold for a positive growth rate in the exponential growth phase, $k>\mu/\alpha$.

For the Bernoulli case, setting the differential equations (\ref{eq:BmeanfieldA})--(\ref{eq:BmeanfieldW}) to 0 gives the equilibrium FRS size:
\begin{equation}
	a_B=\E[A]/n=1-x(1-kx)/(k(1-x)^2).
\end{equation}
Again, the threshold for a positive persistent FRS size at equilibrium coincides with the threshold for positive FRS growth.

\begin{lemma}\label{lem:aBgB} $a_B\leq 0$ iff $g_B\leq 0$.
\end{lemma}
\begin{proof} $k(1-x)^2 a_B=2kx(x-1)+k-x$.  Rearranging $g_B$, $\alpha(1-x)\sqrt{1+4k(1+k)}=\alpha(1-(2k-1)x)+2(1-k)g_b$. So $g_B\leq 0$ iff $(1-x)^2(1+4k(1+k))\leq (1-(2k-1)x)^2$.  Expanding, this is equivalent to $4(1-k)(2kx(x-1)+k-x)\leq 0$ proving the claim.
\end{proof}

The persistent FRS size at equilibrium, like the growth rate, is higher
for larger cross-sectional mean degree $k$ and larger values of the
dimensionless parameter $\alpha/\mu$ (the circle sizes in Figure~\ref{fig:equilplot}). 

If we compare the equilibrium behavior for Poisson degree networks
versus Bernoulli in Figure~\ref{fig:3}, there is a region of parameter values between the
respective thresholds where a persistent FRS is possible for Poisson degree networks,
but not for Bernoulli. For a specific value of the mean degree $k$, this region is between
\begin{equation}
	1/k < \alpha / \mu < 4k/\left(1+2k-\sqrt{1+4k(1-k)}\right)
\end{equation}
which coincides with the region in Figure \ref{fig:2} for positive growth rates.

For all parameter values, the equilibrium persistent size should be larger for Poisson degree distributed networks. Similar to Lemma \ref{lem:gPgB} proving that $g_P\geq g_B$, we can prove for the mean-field models that $a_P\geq a_B$.  This is also illustrated in Figure \ref{fig:3} which shows the steady-state fractions, $a_P$ and $a_B$, as functions of $x$ and $k$. However, we do not see this in simulations, because the simulated networks are not actually homogeneous.

\begin{lemma}\label{aPaB} $a_P\geq a_B$.
\end{lemma}
\begin{proof}
Note that $k(1-x)^2=a_P=x(x-k)+k-x$ and $k(1-x)^2a_B=x(k(2x-1)-k)+k-x$.  Since $x\leq 1/2$, $2x-1\leq 0$, proving the claim.
\end{proof}

The above result assumes that nodes entering the network have the same degree distribution as the rest of the network. In our simulations, nodes enter the network with degree-0, which causes
the network to have a non-homogeneous degree distribution. The degree-0 entry has important interpretations for social networks, and is common for network models with population dynamics \cite{leung_dynamic_2012}. We show that there is a significant effect on the equilibrium size of the FRS. The equilibrium FRS size
increases for Bernoulli networks, and decreases for Poisson networks,
switching their relative positions from our homogeneous model (Figure \ref{fig:traceplot2} and \ref{fig:traceplotadjust}). We make adjustments to our model and derive the FRS outcome in Appendix \ref{sec:meanfield_adjustment}.

\begin{figure}[ht]
    \centering
    \includegraphics[width=\textwidth]{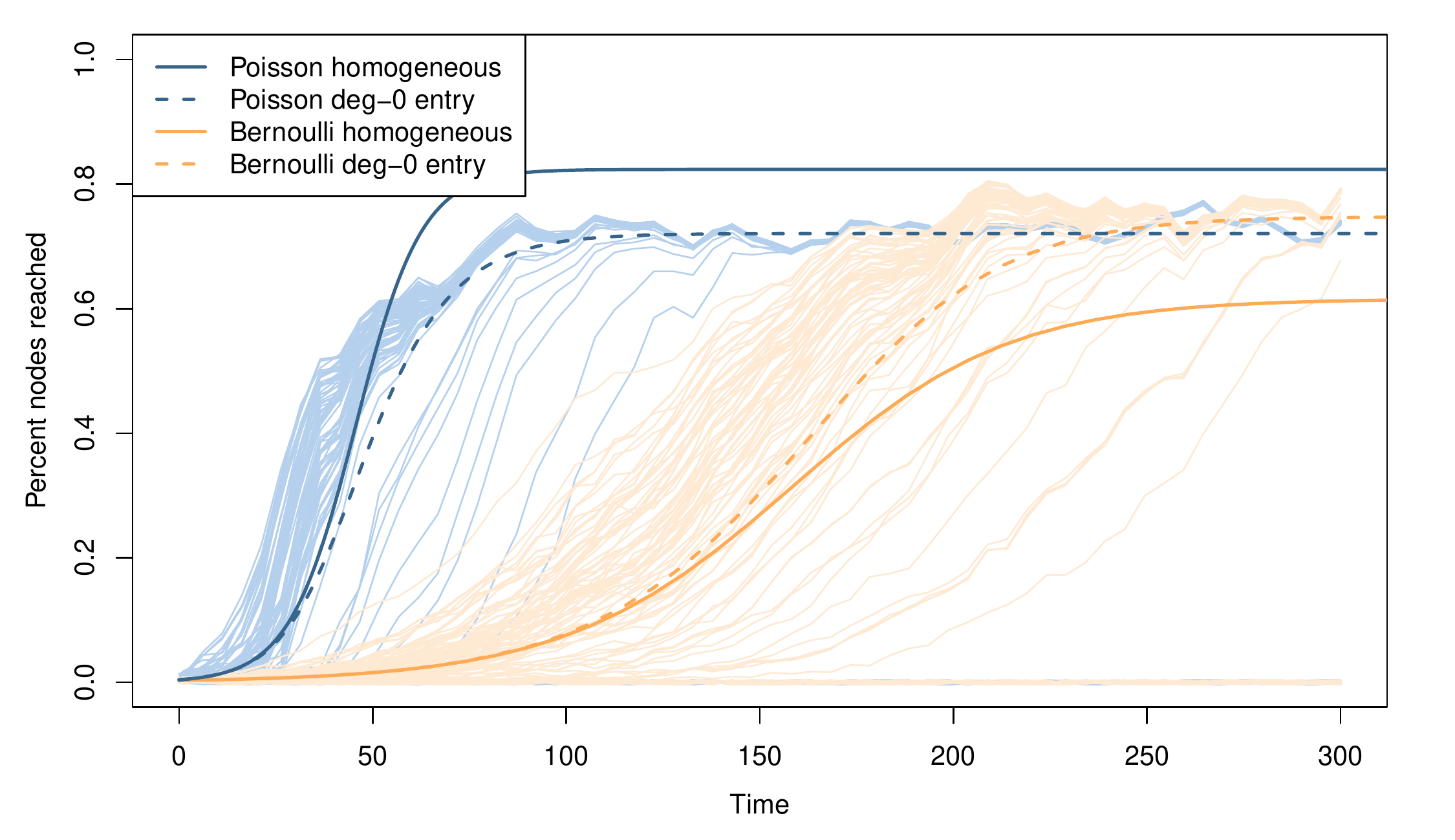}
    \caption{Comparing the persistent FRS size from the homogeneous model (Section~\ref{sec:logistic}) vs the adjusted degree-0 entry model (Appendix \ref{sec:meanfield_adjustment}). For network parameters $k=0.5$, $\alpha=1/6$, $\mu=1/40$, and network size 500 (same simulations as in Figure \ref{fig:traceplot}), the homogeneous model shows higher equilibrium FRS size for Poisson degree networks than Bernoulli. However, in the degree-0 entry model, and in our simulations (thin lines), the equilibrium prevalence for both are nearly the same.}
    \label{fig:traceplotadjust}
\end{figure}

\subsection{Stochastic behavior and extinction}

We can model the FRS extinction as a branching process, a technique
common in epidemiology. Let the population be the nodes in the active FRS. Let the ``offspring'' 
of a node be the non-FRS nodes it contacts, during the time interval when it is in the active FRS. 
Because stochastic extinctions mostly occur within the first few generations, the FRS is small compared to the network, as in the exponential growth phase (Section~\ref{sec:exponential}). So, we can approximate the number of offspring by the cumulative degree of a node, $D_{a}(L)$. This approach is similar to the concept of basic reproduction number $R_0$ in epidemiology \cite{leung_concurrency_2015}. 

Let $G(z)$ be the probability generating
function (PGF) of the offspring distribution $D_a(L)$. Then the probability
of extinction is the smallest non-negative root of the equation \cite{karlin_first_2014}
\begin{equation}
G(z) = z.
\end{equation}

There are two parts to a node's cumulative degree $D_{a}(L)$: the initial neighbors $D_{a}(0)$ at the time it enters the FRS, and the number of formations over its lifetime $D_a(L)-D_a(0)$. The initial degree $D_{a}(0)$ has a distribution specified by the model (either Poisson or Bernoulli). From Section~\ref{sec:aggdeg}
we know that for Poisson degree distributions, the cumulative number of formations
$D_a(L)-D_a(0)$ has a geometric distribution (\ref{eq:DaLPoisson}).

So, the PGF of $D_{a}(L)$ is
\begin{equation}
G(z) = p(1-(1-p)z)^{-1} \exp(k(z-1)),
\end{equation}
where $p = (1+\lambda_{p}/\mu)^{-1}$.
We can solve this equation numerically to find the extinction probability
for different values of the parameters for a Poisson degree network. 
As expected, the extinction probability goes to 1 for 
parameter values below the growth threshold (Section~\ref{sec:keyresults}). We lack the distribution of lifetime number of formations for Bernoulli degree networks for this calculation.

For parameter values
below the threshold, the FRS will go extinct for most initial nodes.
For high values of $k$ and $\alpha/\mu$, the FRS will go to persistence
most of the time. But in a band of parameter values just above the
threshold, the FRS exhibits high variability, where there is large
uncertainty in the equilibrium outcome (Figure~\ref{fig:equilplot}).

\begin{table}[ht]

\begin{tabular}{lcc}
\hline 
 & Poisson & Bernoulli\tabularnewline
\hline 
Edge separation rate & $\sigma=\alpha-2\mu$ & $\sigma=\alpha-2\mu$\\
\noalign{\vspace {.1cm}}
Edge formation rate & $\lambda_{P}=(\sigma+\mu)k$ & $\lambda_{B}=(\sigma+\mu)k/(1-k)$\\
\noalign{\vspace {.1cm}}
Expected lifetime partners & $E[D_{a}(L)]=k\alpha/\mu$ & $E[D_{a}(L)]=k\alpha/\mu$\\
\noalign{\vspace {.1cm}}
Initial FRS growth rate & $g_{P}=\alpha\frac{k-x}{1-k}$ & $g_{B}=\alpha\left(\frac{\sqrt{1+4k(1-k)}-1}{2(1-k)}(1-x)-x\right)$\\
Threshold for growth & $\mu/\alpha<k$ & $\mu/\alpha<1-\frac{2(1-k)}{\sqrt{1+4k(1-k)}+1-2k}$\\
\noalign{\vspace {.1cm}}
Equilibrium prevalence$^{*}$ & $E[A]/n=\frac{k-x}{k(1-x)}$ & $E[A]/n=1-\frac{x(1-kx)}{k(1-x)^{2}}$\\
\noalign{\vspace {.1cm}}
Extinction probability & $z=\frac{p\exp(k(z-1))}{1-(1-p)z}$, solve for $z$ & \\
\hline 
\end{tabular}

\caption{Closed-form expressions for our key results. The expressions depend only on the network parameters $k$ (mean degree), $\alpha$ (edge dissolution rate), and $\mu$ (node exit rate). For simplicity, we defined $x=\mu/\alpha$ and $p=(1+\lambda_{p}/\mu)^{-1}$. The equilibrium prevalence is shown for the base model,
where new nodes are indistinguishable from existing nodes. The prevalence in our simulations, where new nodes enter with degree 0, is derived in Appendix B.}
\label{table:keyresults}
\end{table}

\section{Discussion}
\label{sec:discussion}

We developed analytic expressions and simulation tools for exploring the properties of forward reachable sets in stochastic dynamic networks with open populations.  Our model provides insight into the thresholds, rates of growth, and equilibrium states of the FRS, representing these as emergent features of more basic network properties: degree means and distributions, edge dissolution rates, and node exit rates. All of these are properties that can be measured in sample surveys, which allows our methods to be grounded in empirical data when desired.  We find that the two degree distributions compared here, Bernoulli and Poisson, produce quantitative and qualitative differences in the properties of the FRS when conditioning on all other parameters.  The epidemic threshold is lower and the growth rate is faster in Poisson networks than in Bernoulli networks, for a wide
range of underlying mean degrees and partner acquisition rates.  There is also a broad band of parameter values around the threshold for both network types where the variability in the growth rate and extinction process is maximal.  In this band, the possibility of extinction leads to a bimodal equilibrium FRS size distribution, one that emerges at the beginning of the process, not just over the long term.

A natural application of these findings is in epidemiology, where the FRS represents the maximal epidemic potential in a population. There is a substantial literature on the population dynamics of infectious disease, and a smaller but growing literature on the impact of partnership network structure on these dynamics.  Most of this work relies on deterministic models and mean field approximations for analytic results.  Our work therefore contributes to this literature in several ways.  First, our stochastic analysis describes the variance of the size of the FRS (maximal disease prevalence) on a non-trivial network over time.  Second, we derive analytic solutions for the expected growth rate and growth threshold for this stochastic process.  Third, we are able to derive a mean field approximation for the equilibrium prevalence from our stochastic process of pair formation and dissolution with both vital dynamics and non-trivial network structure; and compare this to stochastic simulations.  And fourth, under these conditions we derive an analytic expression for the extinction probability for the Poisson degree distribution case.  The methods we present are specifically designed to capture the impact of stochastic variability in a systematic way, so that its practical implications can be evaluated. 

There are many different types of infectious diseases; HIV, in particular, motivated our work, and this influenced several aspects of the analysis. We note some of the resulting implications and limitations below.

\subsection{Disease States}
Since nodes do not leave the FRS and nodes cease forming ties after they exit, this model is analogous to an ``SI'' (susceptible-infected) epidemic, and the findings are not relevant for diseases with recovery, either back into susceptible status (SIS) or with immunity (SIR).  Our model is appropriate for HIV, but not for other diseases such as measles or chlamydia.

\subsection{Parameter Values}
We examined a range of parameter values that are both plausible for sexual partnership networks and bound the transitional region for the FRS. 
The mean active lifetime $1/\mu$
is set at 40 years; 35-40 is often used in HIV simulation analyses \cite{baggaley_modelling_2006,hallett_understanding_2008,eaton_why_2014}.
Mean relationship duration $1/\alpha$ ranges from 2.5 to 15 years, and the exponential distribution results in both a large fraction of shorter relationships as well as a right-skewed exponential tail \cite{leung_dynamic_2012}.  Mean cross-sectional degree $k$ ranges from 0.2 to 0.6 partners, similar though a bit lower than that reported in a study of young adults in the U.S. \cite{morris_concurrent_2009}.
Combining these parameters, the average lifetime number of partners ranges from 0.5 to 10, similar to the range reported in many surveys, both in the U.S. and elsewhere \cite{morris_telling_1993,morris_timing_2010,hamilton_consistency_2010}.

\subsection{Degree Distribution Comparison} \label{sec:degdist}
In the context of sexual partnership networks, 
the Bernoulli degree distribution represents a restrictive rule of serial monogamy for partnerships, while the Poisson instead allows individuals to have multiple concurrent partners.  While these distributions do not represent any specific population, they bracket a reasonable range for general heterosexual populations.  Our results provide additional support to the growing literature that finds concurrent partnerships increase both epidemic potential and variability \cite{goodreau_concurrent_2012,morris_concurrent_1997,eaton_concurrent_2011,leung_concurrency_2015}.

There is extensive literature documenting a highly skewed, power-law-like cumulative degree distribution in sexual partnership networks, reviewed in \cite{hamilton2008degree}. However, the degree distributions in our model refer to the number of partners at one moment in time, rather than cumulative over the last year. Although our model would not reproduce a power-law network, we can get closer to a long-tailed distribution by relaxing the homogeneity assumptions in Section \ref{sec:precise}. Specifically if we allow each node to form connections at different rates (but still constant over time), the network will have more nodes with a large cumulative degree, given the same network density. The impact of such heterogeneity on the FRS will be a topic of further investigation.

\subsection{Additional Heterogeneity}
A realistic model would need to account for a wide range of additional demographic details that influence partnership network structure and dynamics:  group-specific mean-degree, dissolution, and mortality rates, and differences in behavior based on disease status. For example, if nodes in the FRS have a reduced partnership formation rate, or an increased mortality rate, then the FRS growth rate will obviously decrease, and the threshold for growth would shift toward larger parameter values for mean degree and dissolution rate.  If these demographic details only affect the dynamic network, our model can account for it by adding more compartments, at the cost of additional equations \cite{ferguson_more_2000}, and the loss of closed-form results.  

Additionally, a memoryless stochastic process (``Markov assumption'') does not account for different types of relationships. In real networks, partnership durations are often bi-modal (long and short term relationships), and partnership durations and type may both depend on the presence of other partnerships for each node. 
As such, our framework seems better suited for gaining broad intuition into the impact of network structure and dynamics on reachability, than for making detailed predictions for specific epidemics. 

\subsection{Partnership status at entry}
In the finite population models and in our simulations, we assume that new nodes enter the network as non-FRS with degree zero (Section~\ref{sec:sim}). As a result, nodes in the FRS have higher mean degree than nodes outside. In the context of disease spread, this would mean that infected individuals will on average have more contacts than the non-infected. The implications are derived in Appendix B.
For Bernoulli networks (i.e., under the assumption of serial monogamy), the implication is that the fraction of singles inside the FRS is lower, especially when the FRS size is large. As a result, a node in the FRS has a higher chance of pairing with an outside node, which leads to a higher rate of discordant contact, and shifts the equilibrium prevalence higher. For Poisson networks, the implications is that the mean degree is lower outside the FRS, leading to smaller components outside the FRS, which reduces the FRS growth rate and shifts the equilibrium prevalence lower. This produces the outcome seen in Figure~\ref{fig:traceplotadjust}: a higher expected equilibrium FRS size for the Bernoulli networks.  If we instead let the new nodes have the same mean degree as the existing ones, the rankings are reversed, and Poisson networks have the higher equilibrium FRS size (as we derive analytically in Lemma~\ref{aPaB}).  We were a bit surprised that this seemingly innocuous assumption would have such an impact on a key outcome of interest.  It suggests that both analytic and simulation studies of epidemic dynamics should pay attention to this assumption in order to avoid inappropriate artifacts.  

\subsection{Reachability vs. Epidemic Spread}
Reachability implicitly defines ``transmission'' as incorporation into the FRS --- it is instantaneous with probability 1 on first contact between discordant nodes (i.e., an infinite transmission rate). This represents the theoretical maximum any transmission process could reach on a specific network. A model for the spread of any specific pathogen would require finite, possibly heterogeneous and possibly time-varying transmission rates across discordant pairs. For example, a model for HIV transmission would need to capture variation in transmission probability by stage of infection, demographic subgroup, treatment status, and the use of prophylaxis. 

Some of the implications of a finite transmission rate are obvious: it would increase the probability of epidemic extinction and the parameter values needed to cross the persistence threshold, and would reduce the growth rate and equilibrium prevalence. There is no reason to expect that the qualitative results for the Bernoulli vs. Poisson comparison would be changed --- a Poisson network will always offer multiple pathways for transmission, and larger connected components for any fixed parameter value set --- but the range of parameter values in which this difference matters will change to reflect the new region of epidemic persistence.

Deriving analytic results for stochastic models with a finite transmission rate, however, will be more complicated.  The epidemic growth rate is proportional to the number of discordant edges, and one would now need to keep track of these over time until transmission actually occurs. Deterministic models that keep track of discordant edges, like the pairwise moment closure model \cite{eames_modeling_2002}, pair-formation model \cite{kretzschmar_effect_1998}, and edge-based compartmental model \cite{miller_edge-based_2011}, offer analytic approaches and some solutions for the growth rate, threshold, and final prevalence of an epidemic.  However, these models often assume that the network is static, that the degree distribution is binomial (equivalent to our Poisson model for large networks) or regular (everyone has the same number of partners), that triad closure is negligible, and all of them ignore stochasticity.  Triad closure is clearly important for finite transmission rates since any node can only be infected once in an SI process, and we have already discussed the importance of the other factors.

\subsection{Stochasticity of Epidemic Outcome}

Most previous papers on transmission processes did not provide any insight into the stochastic variability of epidemic outcomes, or did so for a restrictive model (for example, \cite{dangerfield_integrating_2009} calculated the variance of the outcome on a static regular network).  When infections (like HIV) have a slow transmission rate or networks (like sexual partnership networks) have slow link turnover, then the stochasticity of the initial epidemic growth will be important, and the epidemic outcome cannot be cleanly partitioned by the growth threshold. For parameters above the threshold, there is still a chance that an epidemic can go extinct. Conversely, when the parameters are below the threshold, there is a chance that an epidemic can persist for a long time. 

Our findings suggest that this stochastic variability is important, especially in the band around the persistence/growth threshold. The practical implications are that one may find very different dynamics and outcomes in populations with nearly identical network structure and dynamics. This variability makes it more difficult to infer epidemic potential from observed epidemic outcomes alone, and more difficult to predict epidemic outcomes going forward.  Quantifying this uncertainty, in the context of a finite transmission rate, is an important subject for further investigation.

\appendix
\section{Variance analytics}
We now describe how to determine ODEs for variances, $\Var[X]'$.
Note that $\Var[X]=\E[X^2]-\E[X]^2$ and thus,
$\Var[X]'=\E[X^2]'-2\E[X]\E[X]'$.
More generally, for a vector-valued process,
\begin{align}
\Var[\vect{V}]&=\E[\vect{V}\vect{V}\transpose]-\E[\vect{V}]\E[\vect{V}]\transpose,\\
\Var[\vect{V}]'&=\E[\vect{V}\vect{V}\transpose]'-\E[\vect{V}]'\E[\vect{V}]\transpose
-\E[\vect{V}]{\E[\vect{V}]'}\transpose.
\end{align}
From section~\ref{sec:exponential} we already have equations for $\E[\vect{V}]'$.
We then use the approach in section~\ref{sec:generalODE} to derive equations for $\E[\vect{V}\vect{V}\transpose]'$.
Recall, for jumps of type $i$, $\vect{v}\to\vect{v}+\Delta\vect{V}_i$ at rate $\vect{h}_i\transpose \vect{v}$.  Thus,
\begin{align}
\E[\vect{V}\vect{V}\transpose]' &= \sum_i \E[((\vect{V}+\Delta\vect{V}_i)(\vect{V}+\Delta\vect{V}_i)\transpose-\vect{V}\vect{V}\transpose)\vect{h}_i\transpose\vect{V}]\\
&= \sum_i \E[(\Delta \vect{V}_i\vect{V}\transpose+\vect{V}\Delta\vect{V}_i\transpose+\Delta\vect{V}_i\Delta\vect{V}_i\transpose)\vect{h}_i\transpose\vect{V}].\label{eq:DVV}
\end{align}
This (\eqref{eq:mastereq} and \eqref{eq:DVV}) is still a set of linear differential equations. Specifically, component $(j,k)$ of \eqref{eq:DVV} is a linear combination of an extended set of state variables $\bar{\vect{v}}=(\E[\vect{V}],\E[\vect{V}\vect{V}\transpose])$,
\begin{equation}
\E[\vect{V}\vect{V}\transpose]'_{jk}= 
\sum_i \E[([\Delta \vect{V}_i]_j\vect{V}_k+\vect{V}_j[\Delta\vect{V}_i]_k+[\Delta\vect{V}_i]_j[\Delta\vect{V}_i]_k)\vect{h}_i\transpose\vect{V}].
\end{equation}

\subsection{Poisson case}
Applying this framework to the transitions in the Poisson case (see section~\ref{sec:PoissonEqs}), we derive the following equations:
\begin{align}
\E[A^2]' &= \E[\lambda_P A ((A+C)^2-A^2)]+\E[\mu A((A-1)^2-A^2)]\\
&= 2(\lambda_P \E[C]-\mu)\E[A^2]+(\lambda_P \E[C^2]+\mu)\E[A],
\end{align}
where as before we ignore $R$ since it does not influence the other variables.  Together with the original equation for $\E[A]'$, this is a linear system of differential equations with eigenvalues $g_P$ and $2g_P$.  This reasonably suggests that the standard deviation of $A$ grows at the same rate as the expectation.  Specifically, we have the closed-form solutions
\begin{align}
\E[A^2(t)] &= e^{2g_P t} (\E[A(0)]\E[C^2]\lambda_P+\mu)/g_P \notag\\
  &\qquad+ \E[A^2(0)])-\E[A](\E[C^2]\lambda_P+\mu)/g_P\\
\Var[A(t)] &= e^{2g_P t} (\E[A(0)]\E[C^2]\lambda_P+\mu)/g_P \notag\\
  &\qquad+ \Var[A (0)])-\E[A](\E[C^2]\lambda_P+\mu)/g_P
\end{align}

If we add in the equations relating to $R$,
\begin{align}
\E[AR]' &= \lambda_P \E[A((A+C)R-AR)]+\mu\E[A((A-1)(R+1)-AR)]\\
&=-\mu\E[A]+\mu\E[A^2]+g_P\E[AR]\\
\E[R^2]' &= \mu\E[A((R+1)^2-R^2)]\\
&=\mu(2\E[AR]+\E[A]),
\end{align}
then the largest eigenvalue of the system for $[A,R,A^2,AR,R^2]$ remains $g_P$ as before.

\subsection{Bernoulli case}
We now apply the framework to the transitions for the Bernoulli case (see section~\ref{sec:BernoulliEqs} for a list of the transitions).  We again ignore $R$ as it does not affect the other variables, and derive the following equations:
\begin{align}
\E[A^2]'&=\E[W(2A+1)]\lambda_B+\mu\E[A(-2A+1)]\\
&=\mu\E[A]+\lambda_B\E[W]-2\mu\E[A^2]+2\lambda_B\E[AW]\\
\E[AW]'&=\mu\E[W((A-1)(W-1)-AW)]\notag\\
&\qquad	 +\mu\E[(A-W)((A-1)W-AW)]\notag\\
&\qquad 	+\lambda_B\E[W((A+1)(W-1)-AW)]\notag \\
&\qquad	 +(\sigma/2)\E[(A-W)(A(W+2)-AW)]\\
&= \mu\E[W(-A-W+1)]+\mu\E[(A-W)(-W)]\notag\\
&\qquad 	+\lambda_B\E[W(-A+W-1)]\notag \\
&\qquad 	+(\sigma/2)\E[(A-W)(2A)]\\
&= (\mu-\lambda_B)\E[W]+\sigma\E[A^2]\notag\\
&\qquad 	+(-2\mu-\lambda_B-\sigma)\E[AW]+\lambda_B\E[W^2]\\
E[W^2]'&=(\lambda_B+\mu)\E[W(-2W+1)]+(\sigma/2)\E[(A-W)(4W+4)]\\
&=2\sigma\E[A]+(\lambda_B+\mu-2\sigma)\E[W]\notag\\
&\qquad 	+2\sigma\E[AW]-2(\lambda_B+\mu+\sigma)\E[W^2].
\end{align}
The largest eigenvalue of a 3x3 system involves the solution to a cubic equation and is thus a mess.  Nevertheless, we conjecture that the largest eigenvalue here is less than $2 g_B$ based on numerical experiments.

For completeness, here are the equations involving $R$ (the equation for $\E[R^2]$ is omitted as it is the same as for the Poisson case):
\begin{align}
\E[AR]'&=\lambda_B\E[W((A+1)R-AR)]+\mu\E[A((A-1)(R+1)-AR)]\\
& =-\mu\E[A]+\mu\E[A^2]+\lambda_B\E[WR]-\mu\E[AR]\\
\E[WR]'&=\lambda_B\E[W((W-1)R-WR)]\notag\\
&\qquad 	+\mu\E[W((W-1)(R+1)-WR)]\notag\\
&\qquad 	+(\sigma/2)\E[(A-W)((W+2)R-WR)]\\
& =-\lambda_B\E[WR]+\mu\E[W(W-R-1)]+(\sigma/2)\E[(A-W)2R]\\
& =-\mu\E[W]+\mu\E[W^2]+\sigma\E[AR]-(\lambda_B+\mu+\sigma)\E[WR].
\end{align}

\section{Adjustments to the Mean field Approximation}
\label{sec:meanfield_adjustment}

In our simulations, nodes enter the network with degree-0, which causes
the network to have a non-homogeneous degree distribution. Leung et al. \shortcite{leung_dynamic_2012} showed that the effect on the network is negligible if the partnership dissolution rate is fast compared to the lifetime of each individual node.  However we show that there is still
a significant effect on the equilibrium size of the FRS. The equilibrium FRS size
increases for Bernoulli networks, and decreases for Poisson networks,
switching their relative positions from our homogeneous model (Figure \ref{fig:traceplotadjust}). We make adjustments to
our model and derive the FRS outcome below.

\subsection{Bernoulli degree networks}

In a homogeneous Bernoulli degree network, the probability for
a degree-0 node (``single'') in the FRS to connect with a single
not in the FRS is $(1-E[A]/n)$, where $A$ is the active FRS size,
and $n$ is the network size (Equation \eqref{eq:BmeanfieldA}). If nodes enter the network
as a single (and non-FRS), the percent of singles is higher for the
non-FRS group, compared to the FRS group. So, a single in the FRS
has a higher chance of finding an eligible partner not in the FRS
than in a homogeneous network, and we get a larger equilibrium prevalence
than expected. This effect is negligible when the FRS size is small
compared to the network size, so it will not affect the exponential
growth phase.

We need to keep track of the number of singles not in the FRS ($V$),
and use the conditional probability in the ``Single'' column in
Table \ref{table:2way}, instead of the marginal probability.

\begin{table}
\caption{Two-way table of relationship status vs FRS status}
\begin{tabular}{cccc}
\hline 
 & ~In FRS~ & ~Not in FRS~ & \tabularnewline
\hline 
``Single'' & $W$ & $V$ & $W+V$\tabularnewline
``Paired'' & $A-W$ & $n-A-V$ & $n-W-V$\tabularnewline
\hline 
 & $A$ & $n-A$ & $n$\tabularnewline
\hline 
\end{tabular}
\label{table:2way}
\end{table}

We add $V$ to the state variables of our Markov
process, along with the new transition types in Table \ref{table:adjbernrates}. Note that all deaths in
our simulations are automatically replaced with a new birth as $V$,
to keep the network size constant.

\begin{table}
\caption{New transition types for Bernoulli degree networks}
\begin{tabular}{lcc}
\hline 
Transition & Rate & Change in state (A,W,V)\tabularnewline
\hline 
WV formation (infection) & $W\lambda[V/(V+W)]$ & $(1,-1,-1)$\tabularnewline
WW formation & $W\lambda[W/(V+W)]/2$ & $(0,-2,0)$\tabularnewline
VV formation & $V\lambda[V/(V+W)]/2$ & $(0,0,-2)$\tabularnewline
WW separation & $(A-W)\sigma/2$ & $(0,2,0)$\tabularnewline
VV separation & $(n-A-V)\sigma/2$ & $(0,0,2)$\tabularnewline
W death & $\mu W$ & $(-1,-1,1)$\tabularnewline
V death & $\mu V$ & $(0,0,0)$\tabularnewline
FRS paired death & $\mu(A-W)$ & $(-1,1,1)$\tabularnewline
Non-FRS paired death & $\mu(n-A-V)$ & $(0,0,2)$\tabularnewline
\hline 
\end{tabular}
\label{table:adjbernrates}
\end{table}

We also need an additional differential equation to track changes
in $V$. The modified equations are:
\begin{align*}
\E[A]' & =  \E[W\lambda V/(V+W)-\mu A]\\
\E[W]' & =  \E[-W(\lambda+\mu)+(\sigma+\mu)(A-W)]\\
\E[V]' & =  \E[-\lambda V+(n-A-V)(\sigma+\mu)+\mu n]
\end{align*}
and the initial conditions are $\E[A]=1+k$, $\E[W]=1-k$, $\E[V]=n(1-k)$.
Replace the state variables $A,W,V$ on the right hand side with their expectations to obtain
the deterministic mean-field equations. These equations can be solved
numerically.

\subsection{Poisson degree networks}

In a Poisson degree network, nodes connect to any other node
in the network with equal probability, so we do not have the same
situation as above. However, the mean degree in a network with degree-0
births is different for nodes in the FRS versus those not in the FRS.
Since nodes enter the network as degree-0 and non-FRS, the mean degree
for non-FRS nodes will be lower than for the overall network. For a Poisson
distributed network, lower mean degree leads to lower mean component
size, so we need to adjust Equation \eqref{eq:pois} to the following,
\begin{equation*}
\E[A]' =  (\lambda_{P}\E[C|\mbox{nonFRS}]-\mu)\E[A].
\end{equation*}
The logistic mean-field equation \eqref{eq:PmeanfieldA} can be adjusted in the same way.
We can obtain the mean component size for non-FRS nodes, $\E[C|\mbox{nonFRS}]$,
via simulation, and solve the differential equation as before. Since
this adjustment lowers the number of nodes added to the FRS in a formation
event, it will lower the equilibrium prevalence for a Poisson degree
network.

\bibliography{FRS}

\end{document}